\documentclass[a4paper,12pt]{article} 

\usepackage[a4paper, left=25mm, right=25mm, top=35mm, bottom=35mm]{geometry}
\usepackage[utf8]{inputenc}

\usepackage[english]{babel}
\usepackage[T2A]{fontenc}
\usepackage{indentfirst}

\usepackage{amssymb}
\usepackage{amsthm}
\usepackage[intlimits]{amsmath}

\usepackage{listings}
\renewcommand{\to}{\rightarrow}

\newcommand{\e}{\varepsilon}

\newcommand{\p}{\mathbb{P}}
\newcommand{\E}{\mathbb{E}}
\newcommand{\N}{\mathbb{N}}
\newcommand{\Z}{\mathbb{Z}}
\newcommand{\R}{\mathbb{R}}
\newcommand{\one}{\mathbb{I}}

\newcommand{\be}{\begin{equation}}
\newcommand{\ee}{\end{equation}}
\newcommand{\benn}{\begin{equation*}}
\newcommand{\eenn}{\end{equation*}}
\newcommand{\baln}{\begin{align*}}
\newcommand{\ealn}{\end{align*}}
\newcommand{\bal}{\begin{align}}
\newcommand{\eal}{\end{align}}

\newcommand{\bee}{\begin{eqnarray}}
\newcommand{\eee}{\end{eqnarray}}

\newcommand{\msp}{\mspace{-1mu}}

\newcommand{\ZZ}{\mathrm{ZigZag}}
\newcommand{\StBr}{\mathrm{StBr}}
\newcommand{\ShZZ}{\mathrm{ShortZigZag}}
\newcommand{\Unf}{\mathrm{Unf}}

\newcommand{\RB}{\mathrm{RB}}

\newcommand{\W}{\mathrm{W}}
\newcommand{\B}{\mathrm{B}}
\newcommand{\iB}{\mathrm{iB}}
\newcommand{\iRB}{\mathrm{iRB}}
\newcommand{\wSAW}{\mathrm{SRP}}
\newcommand{\wSAB}{\mathrm{SRB}}

\title{Sub-ballisticity of self-repelling polymers in $\Z^d$.}

\author{Daria Smirnova}
\date{\today}

\begin{document}


%
%
%
%



\addcontentsline{toc}{section}{Bibliography}

\newtheorem{theorem}{Theorem}

\newtheorem{opr}{Definition}
\newtheorem{lemma}[theorem]{Lemma}
\newtheorem{stat}[theorem]{Statement}
\newtheorem{cor}[theorem]{Corollary}
\newtheorem{prop}[theorem]{Proposition}
\newtheorem{ass}[theorem]{Assumption}

\newtheorem{case}{Case}
\renewcommand{\thecase}{\Alph{case}} 

\newtheorem{remark}{Remark}

\renewcommand{\thecase}{\Alph{case}} 


\maketitle

\begin{abstract}
In this article, we prove sub-ballisticity for a class of self-repelling polymers in $\Z^d$. Self-repelling polymers are a two-way generalization of the model of self-avoiding walks, for which the sub-ballisticity was proved by H.~Duminil-Copin and A.~Hammond. Namely, we consider an arbitrary finite symmetric distribution of steps and a more flexible penalization for self-intersections than in the self-avoiding walks model.
\end{abstract}



\section* {Introduction}

\par The model of self-repelling polymers is a probabilistic model defined on a discrete lattice. It is a generalization of the well-known model of self-avoiding walks, in which intersections in a trajectory are allowed at the cost of decreasing the probability of the trajectory. We consider the case where the decrease is computed from a multiplicative coefficient which depends on the number of self-intersections in the trajectory.  Both simple random walks and self-avoiding walks are special cases of self-repelling polymers. We note that, similarly to the model of self-avoiding walks, self-repelling polymers are non-markovian in most cases, in the contrary to the model of simple random walks.

\par self-avoiding walks were introduced by chemist P.~Flory \cite{F49} in the middle of the twentieth century to describe the geometrical shape of polymer chains. Even though in physics a polymer macromolecule is considered in 3-dimensional continuous space with the bond angles usually equal to $\arccos(-1/3)$ (and not to $\pi$ or $\pi/2$), self-avoiding walks on the square lattice can be used as a mathematical model for some aspects of the behavior for the polymer chains. Indeed, Flory predicted universality of the model and the independence of the general behavior of this model with respect to the lattice. This justifies the fact that the model defined on the square lattice is useful for the study of physical polymers. 

\par The similarities between the mathematical model and polymer chains that could be studied experimentally allowed to make many conjectures about the behavior of self-avoiding walks.
Monte-Carlo simulations also give approximate values for several constants of the model and confirmed some of these conjectures, see \cite{S94}.


\par A famous hypothesis \cite{LSW02} states that the distribution of self-avoiding walks on two-dimensional lattices converges to the Schramm-Loewner Evolution of parameter ${{8}/{3}}$. The validity of this conjecture would imply many properties describing the behavior of self-avoiding walks when their length tends to infinity. One of the corollaries would be that for any lattice of dimension at least $2$, the mean-squared distance between the beginning and the end of self-avoiding walks of length $n$ behaves like $n^{2\nu + o(1)}$ with $\nu<1$ \cite{MS96, LSW02}. Note that the latter would imply sub-balisticity, i.e an exponential upper bound on the probability for a self-avoiding walk to go linearly far away from the beginning. Recently, H.~Duminil-Copin and A.~Hammond gave a rigorous proof of sub-balisticity for the lattices $\Z^d$ when $d\geq 2$ \cite{DCH13}.



\par In this paper, we use the method of \cite{DCH13} to extend the sub-ballisticity to a more general model of self-repelling polymers. This model can be used as a better approximation for polymer chains taking into account monomer-monomer connections of different length and a possibility that different parts of the chain can have quite small distance between them. The main result of the paper is that sub-ballisticity holds for this class of models as well.

\par Let us define rigorously the class of self-repelling polymers and state the main result. We start by defining spread-out random walks.


\begin{opr}[Spread-out random walk]
\par Suppose $\Omega$ is a finite subset of vertices of $\Z^d$ which is preserved under the symmetries of $\Z^d$ and does not contain zero. A {\em walk of length $n$} is a sequence $\gamma = (\gamma(i))_{i=0}^n$ of $n+1$ vertices in $\Z^d$ such that $\gamma(i) -\gamma(i-1) \!\in \!\Omega$ for every $0 < i\le n$.
\par The set of walks of length $n$ beginning at $0$ is denoted by $\mathrm{W}_n^\Omega$ (later, we omit the set $\Omega$ in the notation).
The length of $\gamma$ will be denoted by  $|\gamma|$.
\end{opr}


\par The self-repelling polymer is a model of a spread-out random walk with a self-repelling interaction. We follow  \cite{IV08} for the definition.

\begin{opr}[Self-repelling polymer] \label{def}
Consider  $\phi: \Z_+ \to \R_+$ (called the {\em potential}) such that for any $a,\, b \in \Z_+$,
\be \label{phi}
\phi(a+b) \ge \phi(a) + \phi(b)
\ee 
and $\phi(0)=\phi(1)=0$. Note that $\phi$ is necessarily non-decreasing.
Let $l_v(\gamma) = \sum_{k=0}^{|\gamma|} \one_{\gamma(k) = v}$ denote the number of times $\gamma$ visits the vertex $v \in \Z^d$. Let $\rho$ be a jump-distribution on the lattice (i.e. a probability mass function on $\Omega$).
To any $\gamma \in \W_n$,  associate the weight $\sigma(\gamma)$ defined by
\be \label{weight}
\sigma(\gamma) := \biggl(\prod_{v \in \Z^d} e^{-\phi (l_v(\gamma))} \biggr) \biggl(\prod_{i=1}^n \rho (\gamma(i)-\gamma(i-1))\biggr).
\ee
The measure of self-repelling polymers is defined by
\be \nonumber
\forall \gamma_0 \in \W_n , \quad \p_{\wSAW_n}(\gamma_0) = \frac{\sigma(\gamma_0)}{\sum_{\gamma \in \mathrm{W}_n} {\sigma(\gamma)}}.
\ee
\end{opr}
The case $\phi(a) = 0$ for any $a \ge 0 $ corresponds to the classical random walk model. If $\phi(a) = +\infty$ for any $a \ge 2$, then no intersection is allowed and the model corresponds to the self-avoiding walk. For any intermediate potential, intersections of $\gamma$ are allowed but decrease the probability of a walk. The case $\phi(a) = k \cdot (a-1)_+$ is called weakly self-avoiding walks \cite{Sl05}. 


\begin{theorem} \label{MT}
Consider the self-repelling polymer with a jump-distribution $\rho$ which is invariant under the symmetries of the lattice, then
\benn
\lim_{n \to \infty} \tfrac{1}{n} \log \E_{\wSAW_n} \left( |\gamma(n)| \right) < 0.
\eenn
\end{theorem}

\par In Section \ref{sec1}, we extend classical results known for self-avoiding walks to self-repelling polymers regarding the so-called connective constant.  
Then, we prove Theorem \ref{MT} by contradiction in two steps. In Section \ref{secBA}, we show that if Theorem \ref{MT} does not hold then the mean length of a so-called irreducible bridge is finite. In Section \ref{secT2}, we show that the mean length of irreducible bridges is infinite, which altogether with Section \ref{secBA} proves Theorem \ref{MT}.

\par Let us mention two open problems regarding improvements of Theorem \ref{MT} in two directions. 
The first one would be to release the assumption on symmetries for $\rho$. 
The second possible generalization would be to consider arbitrarily large jumps. 
Also, a modification of the proof could give some improvements on the exponential bound obtained in Theorem \ref{MT}.


\section {Preliminaries} \label{sec1}
\par In this section, we extend basic definitions and properties of self-avoiding walks to self-repelling polymers and recall the definition of bridges and irreducible bridges. We introduce the notion of connective constant and prove the analogue of Kesten's lemma in the case of self-repelling polymers.

\par Before doing all of that, we recall a few definitions. Below, $x(v)$ and $y(v)$ denote the first and second coordinates of $v \in \Z^d$.


\begin{opr}
For $\gamma \in \W_n$, the reflection of $\gamma$ under the hyperplane $\{v \in \Z^d, x(v)=0\}$ is denoted $\mathcal R _x(\gamma)$. If $\Z^d$ is invariant under the rotation by an angle $\alpha$, then the clockwise rotation of $\gamma$ around the origin is denoted $r_\alpha(\gamma)$. Note that $\mathcal R _x(\gamma)$ and $r_\alpha(\gamma)$ for $\alpha$ chosen as above belong to $\W_{n}$.


Consider $\gamma_1 \in \W_n$ and $\gamma_2 \in \W_m$. The concatenation $\gamma_1 \circ \gamma_2$ of $\gamma_1$ and $\gamma_2$ is the walk from $W_{n+m}$ defined by 
\be
\gamma_1 \circ \gamma_2 (i) =
\begin{cases}
\gamma_1(i) &\text{ if } 0 \le i \le n, \\
\gamma_1(n) + \gamma_2(i-n) &\text{ if } n \le i \le n+m. 
\end{cases}
\ee
\end{opr}

\subsection{Connective constant for self-repelling polymers}

Let $A$ be a subset of $\W_n$. Introduce 
\be
Z(A) := \sum_{\gamma \in A} \sigma(\gamma).
\ee
If $A = \W_n$, then we simply write ${Z}_n$.
With this notation, for any event $A \subset \W_n,$
\be
\p_{\wSAW_n}(A) = \frac{Z(A)}{Z_n}.
\ee

\begin{theorem}\label{mu_c}
The sequence $(\frac{1}{n}\log Z_n)$ converge. Furthermore, $Z^n \ge e^{\lambda_0 n}$, where
\benn
\lambda_0= \lim_{n \to \infty} \tfrac{1}{n}\log Z_n.\
\eenn
\end{theorem}

\par From now on, we always denote $ \lim \frac{1}{n}\log Z_n$ by $\lambda_0$ and call it the connective constant. 
In the case of self-avoiding walks, $\lambda_0$ is the logarithm of the connective constant of the lattice $\mu_c$. Note that we are usually unable to compute this quantity, except in few cases (for instance, for self-avoiding walks, the connective constant of the hexagonal lattice is known \cite{DCS12} but not the one of the square lattice).

\begin{proof}
\par Any walk $\gamma \in \W_{n+m}$ can be decomposed in a unique way into two walks $\gamma_1 \in \W_{n}$ and $\gamma_2 \in \W_{m}$ so that $\gamma = \gamma_1 \circ \gamma_2$. 
The definition of the potential implies that 
\benn
\phi(l_x(\gamma)) = \phi\bigl(l_x(\gamma_1)+l_x(\gamma_2)\bigr)
 \ge \phi(l_x(\gamma_1))+\phi(l_x(\gamma_2)).
\eenn 
Thus, $\sigma(\gamma) \le \sigma(\gamma_1) \sigma(\gamma_2)$ and 
\be \label{cnadd}
Z_{n+m} \le {
\sum_{	\substack{	\gamma_1 \in \W_n\\ \gamma_2 \in \W_m}	}
\sigma(\gamma_1) \sigma(\gamma_2) } = Z_n Z_m.
\ee
The sequence $\left( \log Z_n\right)$ is therefore sub-additive and non-negative. Fekete's lemma thus implies the existence of a non-negative limit 
\be \nonumber
\lambda_0 = \lim_{n \to \infty} \tfrac{1}{n} \log Z_n
\ee
and the inequality $Z_n \ge e ^{n \lambda_0}$ for all $n$.
\end{proof}
\par Define 
$W (\lambda) = \sum_{n=0}^{\infty} Z_n e^{-\lambda n} $.
It follows from Cauchy-Hadamard Theorem that $W (\lambda)$ converges for any $\lambda > \lambda_0$ and diverges for any $\lambda < \lambda_0$.
Moreover, due to the bound $Z_n \ge e^{\lambda_0 n}$,
 the sum diverges at $ \lambda_0$: 
\be \label{sumwalks}
\sum_{n=0}^{\infty} Z_n e^{-\lambda_0 n} = \infty.
\ee

\begin{opr}[Bridge] \label{RB}
The walk $\gamma$ of length $n$ is called a {\em bridge} if
\be \label{bridgedef}
x(\gamma(0)) < x(\gamma(i)) \le x(\gamma(n)) \quad \text{for } 1 \le i \le n.
\ee
The set of all bridges of length $n$ is denoted $\B_n$ and we set $H_n = Z(\B_n)$.
\end{opr}

%
%
%

\begin{prop} \label{bcdecomp}
\par There exists a constant $C >0$ such that for any $n \in \N$
\be
e^{-C\sqrt{n}} Z_{n} \le H_n \le Z_n.
\ee
\end{prop}

\begin{proof}
\par The second inequality follows directly from Definition \ref{RB}. 
To obtain the first inequality, observe that any walk of length $n$ can be decomposed into bridges by the following procedure.

\par Start with the walks that lie in the upper half-space after the first step, i.e.
\benn
\gamma \in \W_n^+ := \{\gamma \in \W_n,\, x(\gamma(i))>0 \text{ for any } i>0 \}.
\eenn 
Find the latest step $\alpha_1$ such that $(\gamma(i))_0^{\alpha_1}$ is a bridge and cut the walk $\gamma$ at this point: $\gamma = \gamma_1 \circ \gamma_2$. Note that  the value of $\alpha_1=|\gamma_1|$ can be expressed as follows: 
\be \label{alpha1}
\alpha_1 = \max\Big[i: x(\gamma(i)) = \max_{0 \le i \le n} \big(x(\gamma(i))\big) \Big].
\ee
Due to \eqref{alpha1}, all points of $\gamma$ after the $\alpha_1$-th step have a smaller $x$-coordinate than $\gamma(\alpha_1)$. Hence, $\mathcal \gamma_2$ belongs to the reflection of  $\W_{n-\alpha_1}^+$ and therefore can be decomposed using the same method. We continue applying this procedure until the remaining walk is itself a bridge.
\par We obtain the sequence of bridges related to the initial walk. 
Their widths $h_i = | x(\gamma_i (\alpha_i)) - x(\gamma_i(0))|$ are ordered in a strictly decreasing manner and the sum of their lengths $n_i$ is equal to $n$. Also, the height of the walk is bounded by $h_i \le D n_i$, where $D$ is the size of the set $\Omega$ of all possible steps, i.e.
\be \label{D-nlines}
D= \max_{\tilde{\gamma} \in \Omega}\bigl(x(\tilde{\gamma})\bigr).
\ee 
(Note that $D$ is fixed by the definition of the model.)

\par Let us denote $H_{n,h}$ the weight of the set of all bridges of length $n$ and width $h$. The initial walk is uniquely determined by the set of bridges that composed it. Moreover, any set of bridges with decreasing widths corresponds to a walk in $\W^+ = \cup_{n \in \N} \W^+_n$. 
Thus, we can conclude that
\be \label{walk+2bridges}
Z(\W_n^+) \le
 \sum_{\substack{	(n_i)_{i=1}^k \\ \sum_{i=1}^k n_i = n  \\ h_1 < h_2 <\cdots< h_k }	} \prod_{i=1}^k H_{n_i, h_i} \le 
\sum_{\substack{	(n_i)_{i=1}^k \\ \sum_{i=1}^k n_i = n }	} \prod_{i=1}^k H_{n_i} \le 
 H_n \sum_{\substack{	(n_i)_{i=1}^k \\ \sum_{i=1}^k n_i = n }	} 1.
\ee
The last inequality follows from the fact that the composition of brigdes is a bridge and its weight is a product of weights of the initial bridges.
\par It is known that the number of partitions of the integer $n$ is of the order $e^{\tilde{C} \sqrt{n}}$ for some constant $\tilde{C}$ \cite{R18}. 
This, together with \eqref{walk+2bridges}, implies
\be
Z(\W_n^+) \le e^{\tilde{C} \sqrt{n}} H_n.
\ee
\par To extend the proof from $\W_n^+$ to $\W_n$, we cut the walk $\gamma \in \W_n$ at the first point with minimal $x$-coordinate, i.e write $\gamma = \gamma_1 \circ \gamma_2$, where
\be
|\gamma_1| = \alpha_0 =\min\Big[i: x(\gamma(i)) = \min_{0 \le i \le n} \big(x(\gamma(i))\big) \Big].
\ee
Then, the walk $\overline{\gamma_1} = (\gamma_1(\alpha_0-i))_{i=0}^{\alpha_0}$ is a translation of a walk in $\W_{\alpha_0}^+$.
\par The rest of the walk is allowed to visit the initial hyperplane $\{x=0\}$ more than once so to make the decomposition into bridges possible we should add one step at the beginning of the walk. By the symmetry of the set of all possible steps, there exists at least one $\tilde {\gamma} \in \W^+_1$ and 
$\tilde{\gamma} \circ \gamma_2 \in \W_{n-\alpha_0+1}^+$.
\par Both walks $\overline{\gamma_1}$ and $\tilde{\gamma} \circ \gamma_2$ admit the decomposition into bridges as above so the final bound on $Z_n$ is
\be \label{walks2walks+}
Z_n \le \sum_{\alpha=0}^n Z(\W_\alpha^+)  Z(\W_{(n-\alpha+1)}^+) 
\le  e^{2 \tilde{C}  \sqrt{n}} \sum_{\alpha=0}^n H_\alpha H_{n-\alpha+1} \le e^{C \sqrt{n}} H_{n+1}.
\ee
Since $H_{n+1} \ge c H_n$, the result follows.

\end{proof}
%

The previous proposition has the following immediate corollary.
\begin{cor}
\par We have $\lambda_{\text{bridge}} = \lim_{n \to \infty} \tfrac{1}{n} \log H_n = \lambda_0$.
\end{cor}


\begin{cor} \label{RBradius}
The series $H (\lambda) = \sum_{n=0}^{\infty} H_n e^{-\lambda n}$ converges for any  $\lambda > \lambda_0$. Moreover,
\be \label{sumbridges}
\sum_{n=0}^{\infty} H_n e^{-\lambda_0 n} = \infty.
\ee
\end{cor}

\begin{proof}
\par The first part of the statement follows directly from the previous corollary. To prove the second part, one should use the intermediate steps of the proof of Proposition \ref{bcdecomp}. We can use the first inequality of \eqref{walks2walks+} to bound the generating function $W(\lambda)$ as follows:
\bal \label{cntocn+}
\sum_{n=0}^{\infty} Z_n e^{-\lambda n}  
&\le \nonumber \sum_{n=0}^{\infty} \sum_{\alpha=0}^{n}  e^{\lambda n} 
\left( Z(W^+_{\alpha}) e^{-\lambda \alpha} \right)	\left( Z(W^+_{n-\alpha+1}) e^{-\lambda (n-\alpha+1)} \right)
\\ &=  e^{\lambda }  \left( \sum_{n=0}^{\infty} Z(W^+_{n}) e^{-\lambda n} \right)^2\!\!.
\end{align}
\par Due to the first inequality of \eqref{walk+2bridges}, this sum can be rewritten as follows:
\bal  \nonumber
\sum_{n=0}^{\infty} Z(W^+_{n}) e^{-\lambda n}  
&\le \sum_{\substack{	(n_i)_{i=1}^k \\ h_1 < h_2 <\cdots< h_k }	} \prod_{i=1}^k H_{n_i, h_i} e^{-\lambda n_i}  \\ \nonumber 
&=\prod_{h_1 < h_2 <\cdots} \biggl( 1+ \sum_{n=0}^{\infty}  H_{n, h_i} e^{-\lambda n} \biggr) \\ \nonumber
&= \prod_{h_1 < h_2 <\cdots} \exp \biggl(\Bigl( \sum_{n=0}^{\infty}  e^{-\lambda n}  H_{n, h_i} \Bigr)\biggr) \\\nonumber
&= \exp \biggl(\Bigl( \sum_{n=0}^{\infty} (\sum_{h=0}^{\infty}H_{n, h})\Bigr)  e^{-\lambda n} \biggr) \\
&= \exp \biggl(\Bigl( \sum_{n=0}^{\infty} H_n  e^{-\lambda n} \Bigr)\biggr). \label{cn+tobn}
\end{align}
\par 
The divergence of the sum $\sum_{n=0}^{\infty} Z_n  e^{-\lambda n} $ as $\lambda$ tends to $\lambda_0$ (by \eqref{sumwalks}) implies the divergence of the sum $\sum_{n=0}^{\infty} { Z(W^+_{n})} e^{-\lambda n} $ which itself together with \eqref{cntocn+} and \eqref{cn+tobn} implies that $\sum_{n=0}^{\infty} {H_n}  e^{-\lambda_0 n}$ diverges. 
%

\end{proof}

\subsection{Irreducible bridges}

Suppose that $\gamma$ is a bridge of length $n$. 
Then, an integer $i: \in [[ 1, n-1 ]]$ 
is called a {\it{renewal time}} of $\gamma$ if $x(\gamma(i))<x( \gamma(k))$ for any $k>i$ and $x(\gamma(i)) \ge x(\gamma(k))$ for any $k<i$. We sometimes say that $\gamma(i)$ is a \emph{renewal point} if $i$ is a renewal time.
The increasing sequence of all renewal times of the bridge $\gamma$ will be denoted by $R_\gamma$.
The bridge $\gamma$ is called {\it{irreducible}} if $R_\gamma = \emptyset$. The set of all irreducible bridges of arbitrary lengths is denoted $\iB$.

\par A renewal time $r$ splits $\gamma$ into two bridges $(\gamma(i))_{i=0}^{r}$ and $(\gamma(i))_{i=r}^{n}$. From this point of view, an irreducible bridge is a bridge that does not admit any decomposition into shorter bridges.

\begin{theorem} [Kesten's lemma for repelling polymers] \label{keslem}
\be \label{keslemeq}
\sum_{\gamma \in \iB}  e^{ - \lambda_0 |\gamma|}\sigma(\gamma) = 1.
\ee
\end{theorem}

\begin{proof}
\par Each bridge can be seen as a concatenation of $K= |R_\gamma|+1$ irreducible bridges $(\gamma_k)_{k=1}^{K}$. These bridges have no common point except the points $\gamma(R_\gamma(i))$ where one bridge begins and another one ends. Therefore, according to the definition of $\sigma$, we obtain that
\be
\sigma(\gamma) = \prod_{k=1}^{K} \sigma(\gamma_k).
\ee
Thus, we can rewrite the generating function of bridges  in the following way:
\be \label{keslem1}
H (\lambda) =
\sum_{K=0}^\infty \,\,\, \prod_{k=1}^{K} \sum_{\gamma_k \in \iB}\sigma(\gamma_k) e^{-\lambda |\gamma_k|} = 
\frac{1}{1- \sum_{\gamma \in \iB}\sigma(\gamma) e^{-\lambda |\gamma|}}.
\ee
\par From Corollary \ref{RBradius} we obtain that $H (\lambda) $ exists for $\lambda > \lambda_0$ and converges to infinity when $\lambda  \searrow \lambda_0$. Comparing this result with \eqref{keslem1} gives that 
\benn
\lim_{\lambda  \searrow \lambda_0}
\sum_{\gamma_k \in \iB}\sigma(\gamma_k) e^{-\lambda |\gamma_k|} =1.
\eenn 
This equality implies the result of the theorem.
\end{proof}

\par We define a probability measure on the set of irreducible bridges based on Theorem \ref{keslem}. For all $\gamma \in \iB$,
\be
\p_{\iB}(\gamma) = \sigma(\gamma) e^{-\lambda_0 |\gamma|}.
\ee

\par Also, we define the probability measure $\p_{\iB}^{\otimes \N}$ on the set $\B_\infty^+$ of semi-infinite random walks $\gamma: \, \Z_+ \to \left(\Z_+ \right)  \times \Z^{d-1}$ that begin at zero and lie in the upper-half space after the first step by considering a concatenation of irreducible bridges chosen independently according to the probability distribution $\p_{\iB}$, i.e. for any $\gamma \in \B_\infty^+$ and $\tilde{\gamma} \in \iB$,
\be
\p_{\iB}^{\otimes \N}  (\exists \gamma_1 \in \B_\infty^+ \,: \, \gamma = \tilde{\gamma} \circ \gamma_1) = \p_{\iB}(\tilde{\gamma}).
\ee
\par If such $\gamma_1$ exists, then the value $r_1(\gamma) = |\tilde{\gamma}|$ is called the first renewal time of $\gamma$, and all other renewal times are defined by $r_{k+1} (\gamma) = r_k(\gamma_1)$. Also we set $r_0=0$. The sequence of all renewal times $R_\gamma = \left(r_k(\gamma)\right)$ is infinite almost surely.

\par Finite random bridges are related to this model in the following way.
\begin{lemma} \label{infcondfin}
The distribution $\p_{\wSAW_n}$ is equal to 
${\p_{\iB}^{\otimes \N} \big(\cdot \big| \gamma(n) \in R_\gamma \big)}$.
\end{lemma}
\begin{proof}
To prove this lemma, we use almost the same decomposition as in Lemma \ref{keslem}. We have
\begin{align*}
H_n 
&= \sum_{\gamma \in \B_n} \sigma(\gamma) = 
e^{\lambda_0 n} \sum_{K=0}^\infty 
\sum_{\substack{	(\gamma_i \in \iB)_{i=1}^K \\ \sum_{i=1}^K |\gamma_i| = n }	}
\prod_{i=1}^K \sigma(\gamma_i) e^{-\lambda_0 |\gamma_i|} \\
&=\p_{\iRB}^{\otimes \N}\big( \exists K \in \N: \gamma(n) \text{ is a $K$-th renewal point of } \gamma\big) \\
&= \p_{\iRB}^{\otimes \N}\big(\gamma(n) \in R_{\gamma} \big).
\end{align*}
For any event $A$,
we can write the following equality:
\be \nonumber
\p_{\B_n} (A) = \frac{1}{H_n} \sum_{\gamma \in \B_n} \sigma(\gamma) \one_{\gamma \in A} = 
\frac{\p_{\iB}^{\otimes \N} \big((\gamma(i))_{i=0}^n \in A, \,\gamma(n) \in R_{\gamma}\big) }
{\p_{\iB}^{\otimes \N}\big(\gamma(n) \in R_{\gamma} \big)} = 
\p_{\iB}^{\otimes \N} \big(A \,\big| \gamma(n) \in R_\gamma \big).
\ee
\end{proof}

\par The probability measure $\p_{\iB}^{\otimes \N}$ can be extended to bi-infinite bridges $\gamma: \, \Z \to \Z^d$ with the restriction that $\gamma(0)=0$ and is a renewal point, i.e.
\be \nonumber
x(\gamma(-n)) \le 0 \text{ and } x(\gamma(n))>0, \quad \forall n \in \N.
\ee
\par For this type of walks we define the two-sided sequence of renewal points $R_\gamma =(r_k)_{k=-\infty}^{\infty}$ similarly as before.

\par For any renewal point $r_k \in R_\gamma$, define the operation of shift $\tau(\gamma)$ that sets $\gamma(r_1)$ to zero: $\tau(\gamma) = (\gamma(i-r_1))_{i=-\infty}^{\infty}$. By construction, the probability measure $\p_{\iB}^{\otimes \Z}$ is invariant under $\tau$.
\begin{lemma} \label{tauerg}
$\p_{\iB}^{\otimes \Z}$ is ergodic under $\tau$.
\end{lemma}
\begin{proof}
\par Let $A$ be a shift-invariant measurable event. 
Then, for any choice of $\e>0$, we can pick a positive integer $M$ and an event $A_M$ depending only on $(\gamma(i))_{i=-M}^M$ such that $\p_{\iRB}^{\otimes \Z}(A \bigtriangleup A_M) \le \e$. 
We can express $\p_{\iRB}^{\otimes \Z} (A)$ as $\p_{\iRB}^{\otimes \Z}(A \cap A) = \p_{\iRB}^{\otimes \Z}(A \cap \tau^{4M}(A))$. 
This probability is bounded in the following way:
\bal \nonumber
\! \left| \p_{\iRB}^{\otimes \Z}(\msp A \msp \cap \msp \tau^{4M}\!(A))\msp \! - 
\msp \!\p_{\iRB}^{\otimes \Z}(\msp A_M \msp \cap \msp \tau^{4M}\!(A_M)) \right| 
\! &\le \p_{\iRB}^{\otimes \Z} \! \left((A \cap \tau^{4M}(A))\!\bigtriangleup\! (A_M \cap \tau^{4M}(A_M))\msp \right) \\ \nonumber
   &\le \p_{\iRB}^{\otimes \Z}	(\msp A \!\bigtriangleup \!A_M)\! 
\msp +\p_{\iRB}^{\otimes \Z}	(\msp \tau^{4M}\msp \msp (A)\!\bigtriangleup \!\tau^{4M}\msp \msp (\msp \msp A_M)\msp ) \\ 
   &\le  2 \e. \qquad \label{tauerg1}
\end{align}
All irreducible pieces of $\gamma$ are sampled independently and $r_{2M} \ge 2M$ so the events $A_M$ and  $\tau^{4M}(A_M)$  depend on the independent collections of irreducible bridges $(\gamma_m)_{m=-M}^{M}$ and $(\gamma_m)_{m=3M}^{5M}$ so we can write the following estimation:

\begin{align} \nonumber
\left| \p_{\iRB}^{\otimes \Z}(A_M \! \cap\msp  \tau^{4M}\!(A_M) \!-\!  \p_{\iRB}^{\otimes \Z}(A)^2\msp)\right| 
&\! \le \msp \left|  \p_{\iRB}^{\otimes \Z}(A)^2 \!+ \!2\e   \p_{\iRB}^{\otimes \Z}(A) \! + \!\e^2 \!-  \p_{\iRB}^{\otimes \Z}(A)^2)\right| \\
&\!  \le  3\e .\label{tauerg2}
\end{align}
The bounds \eqref{tauerg1} and \eqref{tauerg2} give that
\be
|\p_{\iRB}^{\otimes \Z}(A)-\p_{\iRB}^{\otimes \Z}(A)^2 | \le 5\e
\ee
which is true for any choice of $\e$. This implies that $\p_{\iRB}^{\otimes \Z}(A) \in \{0,1\}$.
\end{proof}


\section 
{Ballistic assumption} \label{secBA}
Both in this section and in the next one we prove Theorem \ref{MT} by disproving its contrary. The slight modification of this contrary will be called the Ballistic Assumption. Note that these modifications do not change the validity of the statement.


\begin{ass}[Ballistic assumption] \label{BAprop}
There exists $v >0$ such that:
\be \label{BA}
\limsup_{n \to \infty} \tfrac{1}{n} \log \p_{\wSAB_n}(x(\gamma(n))> v n) = 0,
\ee
where $\p_{\wSAB_n}$ is a measure defined as in Definition \ref{def} but on the set of bridges of length $n$.
\end{ass}

\par We work not with the whole set $Z_n$ of self-repelling polymers of length $n$, but only a subset $H_n$ of self-repelling bridges. 
 To justify this change we first need to observe that for self-avoiding walks a linear lower bound on the distance from the origin implies a linear lower bound on at least one coordinate of the endpoint. 
 The second fact confirming the identity of these two assumptions is that 
\be
e^{-C\sqrt{n}}\p _{\wSAB_n} ( x(\gamma(n))> v n) 
\le \p _{\wSAW_n} (x(\gamma(n))> v n) 
\le e^{C\sqrt{n}}\p _{\wSAB_n} ( x(\gamma(n))> v n).
\ee
This inequalities follow from Theorem \ref{bcdecomp} and the same decomposition as in Theorem \ref{bcdecomp} applied to the sets $\{ \gamma \in \wSAB_n,  x(\gamma(n))> v n)\}$ and $\{ \gamma \in \wSAW_n,  x(\gamma(n))> v n)\}$.
\par Let us investigate some consequences of Assumption \ref{BAprop}.
Let us work only with bridges wide enough to be involved in \eqref{BA}: $\RB_{n,v} = \big\{ \gamma \in \RB_n: x(\gamma_n)> v n \big\} $.
The Ballistic Assumption puts the following restriction on the number of renewal points.

\begin{theorem} \label{T1}
If \eqref{BA} holds, then for any increasing sequence of positive integers $(u_n)$, there exists $\delta >0$ such that
\be \label{T1eq}
\limsup_{n \to \infty} \tfrac{1}{u_n} \log \p_{\wSAB_{u_n,v}}(|R_{\gamma}|> \delta u_n) = 0.
\ee
\end{theorem}

\par To prove this theorem, we generalize the idea of renewal points and look at the hyperplanes $\pi_{x_0+1/2}=\{ v \in R^d, x(v)=x_0+1/2\}$ that have not many crossings with segments, corresponding to the steps of $\gamma$:
\benn
Rl_{\gamma}^m = \Big\{ x_0\in \N: 1 \le \Big|\{i: (\gamma(i), \gamma(i+1)) \cap \pi_{x_0+1/2} \ne \emptyset \} \Big| \le m \Big\}.
\eenn
where $(\gamma(i), \gamma(i+1))$ denotes the segment between the points $\gamma(i)$ and $\gamma(i+1)$.
It is easy to see that $Rl_{\gamma}^m \subset Rl_{\gamma}^{m+1}$ for any positive $m$.
\par For nearest-neighbor walks, there is a bijection between $R_{\gamma}$ and $Rl_{\gamma}^1$. This is not true in the more general case, but nonetheless there exists $D>0$ such that
\be \label{R2}
|R_{\gamma}| \le |Rl_{\gamma}^1| \le D |R_{\gamma}|
\ee
where $D$ is defined as in \eqref{D-nlines} and depends only on $\Omega$.

\par Define the following subsets of $\RB_{n,v}$:
\benn
\RB_{n,v,\delta }^m = \{\gamma \in \RB_{n,v} : \big| Rl_{\gamma}^m \big| > \delta n\}.
\eenn

Theorem \ref{T1} is a consequence of the following lemma.
\begin{lemma} \label{L1}
If the Ballistic Assumption holds, then for any  $v, \, \delta>0$ and $m\ge2$ and 
for any sequence $(u_n)$ in $Z_+$,
there exists $\delta' >0$ and a subsequence $(t_n)$ of $(u_n)$ such that
\be \label{L1eq}
\limsup_{t_n \to \infty} \tfrac{1}{t_n} \log \left( \frac{Z(\RB_{t_n,v,\delta }^m)}{Z(\RB_{t_n,v,\delta' }^{m-1})} \right) \le 0.
\ee
\end{lemma}

\par The proof of this lemma is based on the unfolding operation defined below.

\begin{opr}
Fix $\gamma \in \RB_n$. The pair of integers $(i,j), \, 0 < i \le j <n $ is called a {\em zigzag} of $\gamma$ if
\bal
& x(\gamma(k)) \le x(\gamma(i)) \quad &\forall&\, k<i, \nonumber \\
x(\gamma(j)) \le \,  & x(\gamma(k))  < x(\gamma(i))\quad &\forall& \, i<k<j, \nonumber \\
x(\gamma(j)) <  \, & x(\gamma(k)) \quad &\forall&\, k>j. \nonumber
\end{align}
\par The set of all zigzags of the walk will be denoted by $\ZZ_\gamma$.
\end{opr}
\par All zigzags in the set $\ZZ{_\gamma}$ are disjoint.

\begin{opr}[Unfolding]
Suppose that $\gamma \in \RB_n$ and $(i,j) \in \ZZ_{\gamma}$. Then, define the new bridge:
\be
\Unf_{(i,j)}(\gamma) := (\gamma(k))_{k=0}^i \circ \mathcal{R}_x\big((\gamma(k))_{k=i}^j\big) \circ (\gamma(k))_{k=j}^n.
\ee 
\end{opr}
Let us recall some elementary properties of this operation. 
\begin{lemma}
The following properties hold for any bridge $\gamma \in \RB_n$ and for any $(i,j) \in \ZZ_{\gamma}$:
\begin{itemize}
\item $\sigma(\gamma) \le \sigma \big(\Unf_{(i,j)}(\gamma)\big)$,
\item $x(\gamma(n)) \le x \big((\Unf_{(i,j)}(\gamma))(n)\big)$,
\item $\gamma(i)\in R_{\Unf_{(i,j)}(\gamma)}, \, \big(\Unf_{(i,j)}(\gamma)\big)(j) \in R_{\Unf_{(i,j)}(\gamma)}$,
\item $\ZZ_{\gamma} \backslash \{i,j\} \subset  \ZZ_{\Unf_{(i,j)}(\gamma)} $,
\item $\forall (i,j) , (\tilde{\imath} ,\tilde{\jmath}) \in \ZZ_{\gamma}: \, 
\Unf_{(i,j)} \Unf_{(\tilde{\imath}, \tilde{\jmath})}(\gamma) = \Unf_{(\tilde{\imath} ,\tilde{\jmath})}\Unf_{(i,j)}(\gamma) .$
\end{itemize}
\end{lemma}
We do not include the (easy) proof of this statement.
The last property allows us to define the unfolding of a set of zigzags as a row of successive unfoldings (the order in which we do the unfolding operations is irrelevant).

\begin{proof}[Proof of Lemma \ref{L1}]
Let us fix $m>2$, $v>0$, $\delta>0$ and any sequence of positive integers $(u_n)$. Look at the sequence $( \RB_{u_n,v,\delta }^m )^\infty_{n=0}$.
At least one of the following propositions must be true:

\begin{case} \label{case1}
The number of sets where a positive density of renewal points has quite high probability to occur is infinite:
\par There exists $\delta'>0$ and a subsequence $(t_n)$ of $(u_n)$ such that
\be \label{case1eq}
Z\bigl( \bigl\{ \gamma \in   \RB_{t_n,v,\delta }^m : |R_\gamma| \ge \delta' t_n  \bigr\}\bigr) \ge \tfrac{1}{3} Z\bigl(  \RB_{t_n,v,\delta }^m \bigr).
\ee
\end{case}

\begin{case} \label{case2}
In any set $\RB_{u_n,v}$ there is a good probability that the number of zigzags in a walk is sufficiently small:
\par There exists $(\e_n) \searrow 0$ such that
\be
Z  \bigl( \{ \gamma \in \RB_{u_n,v,\delta }^m :  |\ZZ_\gamma| \le \e_n u_n \} \bigr) 
\ge \tfrac{1}{3} \, Z (\RB_{u_n,v,\delta }^m).
\ee
\end{case}

\begin{case} \label{case3}
The number of $u_n$ such that a bridge with positive density of zigzags and sufficiently small number of renewal points has a high probability to occur in $\RB_{u_n,v}$ is infinite.

\par There exist $\e>0$,  a sequence $(\delta'_n) \searrow 0$ and a subsequence $(t_n)$ of $(u_n)$ such that
\be \label{case3eq}
Z  \bigl( \{ \gamma \in \RB_{t_n,v,\delta }^m : |R_\gamma| \le \delta'_n t_n \text{ and }|\ZZ_\gamma| \ge 2\e t_n \} \bigr) 
\ge \tfrac{1}{3} \,Z (\RB_{t_n,v,\delta }^m).
\ee
\end{case}

\begin{proof}[Proof in Case \ref{case1}]
\par Inequality \eqref{L1eq} follows directly from \eqref{case1eq} and the fact that \newline $|Rl_{\gamma}^{m-1}| \ge |Rl_{\gamma}^{1}| \ge |R_{\gamma}|$ for any $m$ by \eqref{R2}.
\end{proof}
%

\begin{proof}[Proof in Case \ref{case2}]
\par Take $\gamma \in \RB_{u_n,v,\delta }^m$ satisfying the property $|\ZZ_\gamma| \le \e_n u_n$.
\par For each hyperplane $\pi \in Rl_{\gamma}^m$ that has exactly $m$ crossings with $\gamma$, there exists at least one zigzag $(j,k) \in \ZZ_\gamma$ such that $(\gamma(i))_{i=j}^k$ intersects $\pi$. 
Hence, the two parts $(\gamma(i))_{i=0}^j$ and $(\gamma(i))_{i=k}^{u_n}$ also have at least one crossing with $\pi$. 
If we unfold this zigzag, then all points of $\gamma$ after step $j$ will have a larger $x$-coordinate than in $\pi$. 
Then, $\pi$ has no more than $m-2$ crossings with $\Unf_{(i,j)}(\gamma)$. 
For any other hyperplane in $Rl_{\gamma}^m$, there is a corresponding hyperplane in $Rl_{\Unf_{(j,k)}(\gamma)}^m$ with at most the same number of crossings. 
\par Let us repeat this operation and unfold all zigzags in $\ZZ_\gamma$. The resulting walk $\tilde{\gamma} = \Unf_{\ZZ_{\gamma}}(\gamma)$ will satisfy the following property:
\be
|Rl^{m-1}_{\tilde{\gamma}}| \ge |Rl^{m}_{\gamma}|.
\ee
\par Now, let us choose a walk $\tilde{\gamma} \in \RB_{u_n,v,\delta}^{m-1}$ and bound the number of walks $\gamma$ that gives $\tilde{\gamma}$ as a result of the unfolding operation. The number of $ \gamma \in \RB_{u_n,v,\delta }^m$ such that $ |\ZZ_\gamma| \le \e_n u_n $ and $ \Unf_{\ZZ_{\gamma}}(\gamma) = \tilde{\gamma}$ is equal to the number of possible ways to pick at most $2 \e_n u_n$ points of $\tilde{\gamma}$ 
to form all zigzags in $\gamma$. Thus,
\be \label{unf-1card}
|\{\gamma \in \RB_{u_n,v,\delta }^m: \Unf_{\ZZ_{\gamma}}(\gamma) = \tilde{\gamma} \}|
\le \sum_{k=0}^{2e_n u_n} {u_n \choose k} 
\le \exp{\left(2 u_n \e_n \log\left(\tfrac{1}{2 \e_n}\right)\right)}.
\ee
Here, we used the bound
$\sum_{k=0}^{\e n} {n \choose k} \le 2^{nH(\e)}$, where $\e < 1/2$ and $H(\e) = -\e \log_2(\e) - (1-\e)\log_2(1-\e)$.
\par For all $\gamma$ in this set, $\sigma(\gamma) \le \sigma(\tilde{\gamma})$ so
\benn
\sigma(\tilde{\gamma}) \ge \sum_{\gamma \in \RB_{u_n,v,\delta }^m, \,\Unf_{\ZZ_{\gamma}}(\gamma) = \tilde{\gamma}}
\frac{\sigma({\gamma})}{\exp{\left(2 u_n \e_n \log\left(\frac{1}{2 \e_n}\right)\right)}}.
\eenn

\par The set of all $\tilde{\gamma}$ that have a preimage in $\RB_{u_n,v,\delta }^m$ is not bigger than $\RB_{u_n,v,\delta }^{m-1}$, so
\be
Z\left( \RB_{u_n,v,\delta}^{m-1} \right) \ge 
\frac{\frac{1}{3} Z\left( \RB_{u_n,v,\delta}^{m} \right)}{\exp{\left(- 2 u_n \e_n \log\left(\frac{1}{2 \e_n}\right)\right)}} .
\ee
This inequality and the fact that $(\e_n) \searrow 0$ implies the statement of Lemma \ref{L1}.
\end{proof}

\begin{proof}[Proof in Case \ref{case3}]
\par The idea of this proof is to unfold the necessary number of small zigzags and to obtain some renewal points by this unfolding.

\par Let us take a bridge $\gamma \in \RB_{t_n,v,\delta }^m$ such that $|R_\gamma| \le \delta'_n t_n$ and $|\ZZ_\gamma| \ge 2\e t_n$. 
We can define a set containing all small zigzags of $\gamma$: 
\be
\ShZZ_\gamma = \left\{(i,j) \in \ZZ_\gamma: j-i \le \tfrac{1}{\e}\right\}.
\ee
The central sections of all zigzags in $\ZZ_\gamma$, i.e. the parts $(\gamma(k))_{k=i}^{j}$, are disjoint and the sum of the number of steps in all central sections is not bigger than $t_n$. 
Inequality $|\ZZ_\gamma| \ge 2\e t_n$ implies that 
\be
\left\{\gamma\in \RB_{t_n,v,\delta }^m: \left|\ZZ_\gamma \right| \ge 2\e t_n \right\} \subset
\left\{\gamma\in \RB_{t_n,v,\delta }^m: \left|\ShZZ_\gamma \right| \ge \e t_n \right\}.
\ee
Let us define
\be
RB_c =  \{ \gamma \in \RB_{t_n,v,\delta }^m : |R_\gamma| \le \delta'_n t_n \text{ and }|\ShZZ_\gamma| \ge \e t_n \}.
\ee
It is easy to see that this set contains the set defined on the right-hand side of \eqref{case3eq}.

\par Now, let us take a subset $\mathrm{ZZ} \subset \ShZZ_\gamma$ of size $\e' t_n < \e t_n$. The precise value of $\e'$ will be defined later. Then, unfold all zigzags in $\mathrm{ZZ} $. The resulting walk  $\tilde{\gamma} = \Unf_{\mathrm{ZZ} }(\gamma)$ has at least $2 \e' t_n$ renewal points, i.e.
\be
\tilde{\gamma} \in \RB_{t_n,v, 2 \e' }^1.
\ee
Different walks can be obtained by the different choice of $\mathrm{ZZ} $.  
For fixed $\gamma$, the number of ways to pick $\mathrm{ZZ} $ can be estimated as follows:
\be \label{min_gt_from_g}
\left|\left\{
\tilde{\gamma}: \exists \mathrm{ZZ}  \subset \ShZZ_\gamma, |\mathrm{ZZ} | = \e' t_n, \Unf_{\mathrm{ZZ} }(\gamma) = \tilde{\gamma}
\right\}\right| \ge \binom{\e t_n}{\e' t_n}.
\ee

Define the set of all possible pairs $(\gamma, \mathrm{ZZ} )$:
\be
\mathrm{BZ}  = \left\{ \left(\gamma,\mathrm{ZZ}  \right): 
\gamma \in RB_c, \mathrm{ZZ}  \subset \ShZZ, |\mathrm{ZZ} | = \e' t_n \right\}\!.
\ee

The number of renewal points of $\tilde{\gamma}$ can be bounded from below.

\par Let us unfold one zigzag $(i,j)\in \mathrm{ZZ} $ in $\gamma$ and look at the number of crossings of $\gamma$ and $\tilde{\gamma} = \Unf_{(i,j)}(\gamma)$ with different hyperplanes $\pi_{x_0} = \{x= x_0+ 1/2\}$.
For any $x_0 < x(\gamma(j))$, the number of crossings is preserved, so $\tilde{\gamma}$ does not contain any renewal points except the points that were already present in $R_\gamma$. 
For any $x_0 \ge x(\gamma(i))$, there is a correspondence between the crossings of $\gamma$ and $\pi_{x_0}$ and between the crossings of $\gamma$ and $\pi_{x_0}$ and the crossings of $\tilde{\gamma}$ and $\pi_{x_0+2(x(\gamma(i))-x(\gamma(j)))}$. This part of $\tilde{\gamma}$ will nor have any new renewal points. 
The remaining middle part of $\tilde{\gamma}$ has width $3\left|x(\gamma(i))-x(\gamma(j))\right|$ that can be bound by $3(j-i)D$, where $D$ is defined in \eqref{D-nlines}. In this gap there can be maximum $3(j-i)D$ renewal points. Note that $(i,j) \in \ShZZ_{\gamma}$ and that $j-i\le \frac{1}{\e}$.

\par This operation can be applied consequentially for all zigzags in $\mathrm{ZZ} $ and gives the following result:
\be \label{NRenewal}
|R_{\Unf_{\mathrm{ZZ} }(\gamma)}| \le \delta'_n t_n + \e' t_n \frac{3 D}{\e}.
\ee

\par For each $\tilde{\gamma}$, there can be many pairs $(\gamma, \mathrm{ZZ} ) \in \mathrm{BZ} $ that gives $\tilde{\gamma}$ after unfolding. Their number can be bounded in the following way.
\par The number of possible ways to make $\e' t_n$ zigzags to obtain $\gamma$ from $\tilde{\gamma}$ is not bigger than the number of ways to choose $2 \e' t_n$ points from all renewal points of $\tilde{\gamma}$. Then, we can use  inequality \eqref{NRenewal} to obtain the following bound: 
\be \label{max_g_from_gt}
\left| \left\{
(\gamma, \mathrm{ZZ} ) \in \mathrm{BZ}  :\, 
\Unf_{\mathrm{ZZ} }(\gamma) =\tilde{\gamma}
\right\} \right| \le
\binom{\delta'_n t_n + \e' t_n \frac{3 D}{\e} }{ 2 \e' t_n}.
\ee

\par We can use inequalities \eqref{case3eq}, \eqref{min_gt_from_g} and \eqref{max_g_from_gt} and the bound 
$\left(\frac{a}{b}\right)^b \le \binom{a}{b} \le \left(\frac{e \cdot a}{b}\right)^b$ to obtain
\begin{align}
Z ( \RB_{t_n,v, 2 \e' }^1) 
&\ge Z ( RB_c ) \cdot \frac
{\min_{\gamma \in RB_c} \left|\left\{ \tilde{\gamma}: \exists (\gamma, \mathrm{ZZ} ) \in \mathrm{BZ} , \Unf_{\mathrm{ZZ} }(\gamma) = \tilde{\gamma}\right\}\right|}
{\max_{\tilde{\gamma} \in \RB_{t_n,v, 2 \e' }^1}\left|\left\{ (\gamma, \mathrm{ZZ} ) \in \mathrm{BZ} : \Unf_{\mathrm{ZZ} }(\gamma) = \tilde{\gamma}\right\}\right|} \nonumber \\ \nonumber
&\ge  \frac{1}{3} Z ( \RB_{t_n,v,\delta }^m ) \cdot \frac
{\binom{\e t_n}{\e' t_n}}
{\binom{\delta'_n t_n + \e' t_n \frac{3 D}{\e} }{ 2 \e' t_n}} \\
& \ge \frac{1}{3} Z ( \RB_{t_n,v,\delta }^m ) \cdot 
\left( \frac{\e}{\e'}\cdot \frac{4}{e^2} \cdot \left(\frac{\e'}{\delta'_n + 3 D\frac{\e'}{\e}}\right)^2\right)^{ \!\!\! \e' t_n}\!\!\!\!\!\!\!.\,
\end{align}
\par The result of Lemma \ref{L1} holds when a constant lower bound on 
$\frac{\e \cdot \e'} {\left(\delta'_n + 3 D\frac{\e'}{\e}\right)^2}$ is bigger than $2$. We can choose 
$\e' = \frac{\e^3}{2(6D)^2}$ and use the fact that $\delta'_n \searrow 0$ to obtain this bound. Then,
\benn
Z ( \RB_{t_n,v, 2 \e' }^1) \ge \frac{1}{3} Z ( \RB_{t_n,v,\delta }^m ) \left( \frac{8}{e^2} \right)^{\e' t_n},
\eenn
which implies \eqref{L1eq}.



\end{proof}
This finishes the proof of Lemma \ref{L1} in the general case.
\end{proof}

\begin{proof}[Proof of Theorem \ref{T1}]
For any $\gamma \in \RB_{n,v}$ the number of hyperplanes $\pi_k$ crossed at least once is bigger than $v n$. Thus, at least half of them are crossed less than $\frac{2 D}{v}$ times where $D$ is defined in \eqref{D-nlines}. Hence, we deduce that
\be 
\RB_{u_n,v} = \RB_{u_n,v, \frac{v}{2}}^{\frac{2 D}{v}}.
\ee
\par We can apply Lemma \ref{L1}  ${\frac{2 D}{v}}$ times with $m$ chosen decreasingly from ${\frac{2 D}{v}}$ to $2$. We have to take $\delta = \frac{v}{2}$ at the first step and set it equal to $\delta'$ from the previous step afterwards. Then, we can find $\delta'>0$ and $(t_n$ a subsequence of $(u_n))$ such that
\be \nonumber
\lim_{t_n \to \infty} \frac{1}{t_n} \log  \left( \frac{Z(\RB_{t_n,v,\delta' }^1)}{Z(\RB_{t_n,v, \frac{v}{2}}^{\frac{2 \alpha}{v}} )} \right) \ge 0.
%
%
\ee
\par The proof follows for $\delta'' = \tfrac{\delta'}{D}$.
\end{proof}

\begin{cor} \label{collsec2}
If the Ballistic Assumption holds, then 
\be \label{Einf}
\E_{\iRB}(|\gamma|) < \infty.
\ee
\end{cor}
\begin{proof}

Suppose that $\E_{\iRB}(|\gamma|) = \infty$. This implies that for any choice of $C >0$ and $\alpha > 0$, there exists a bound $x_0(C,\alpha)$ such that for any $x>x_0$,
\be \nonumber
\p_{\iRB}(|\gamma| \ge x) > \frac{C}{x^{1+\alpha}}.
\ee
Let us fix $C_0 =8/9$ and $\alpha_0=1/2 $ and define $M = \max(\frac{2}{v}, x_0(C_0,\alpha_0))$.

For any positive constant $A$, we can construct a three-point distribution $X$ as follows:
\bal
\p_X (X=0) &= \p_{\iRB}(|\gamma|< M), \, \nonumber \\
\p_X (X=M)& = \p_{\iRB}( M \le |\gamma| < A),\, \nonumber \\
\p_X (X=A) &=  \p_{\iRB}(|\gamma|\ge A). \nonumber
\end{align}

\par The expectation of this distribution is not smaller than $M$ if for a pair $(C,\alpha)$ satisfying $x_0(C, \alpha) \le M$:
\be \label{MAineq} 
(2M)^\alpha \le A^\alpha \le \frac{C}{2} \frac{M^\alpha}{M^{1+\alpha}-C}.
\ee
Futhermore these two inequalities hold simultaneously if 
\be \nonumber
C \ge \frac{(2M)^{1+\alpha}}{1+2^{1+\alpha}}.
\ee
\par Let us choose $\alpha = 2$ and correspondingly $C = \frac{8}{9} M^{3}$. Then, $\frac{C}{x_0^{1+\alpha}} = \frac{C_0}{x_0^{1+\alpha_0}}$ implies that $x_0(C, \alpha) \le M$ and allows to choose $A$ accordingly to \eqref{MAineq} and obtain
\be \label{expX}
\E_X (X) \ge \tfrac{2}{v}.
\ee

\par The size of the set of renewal points can be estimated as follows:
\be
\p_{\RB_n} (|R_\gamma| > v n) = \p_{\iRB} (\sum_{i=1}^{v n} |\gamma_i| < n) \le \p_X \left(\tfrac{1}{vn}\sum_{i=1}^{vn} X_i < \tfrac{1}{v}\right),
\ee
where random variables $X_i$ are independently distributed according to $X$.
\par Because of \eqref{expX}, this probability can be estimated by Cramer's Theorem \cite{C38} in large deviation theory. There exists a positive constant $c>0$ such that for any $n$ large enough
\be
\p_{\RB_n} (|R_\gamma| > v n) \le \p_X \left (\tfrac{1}{vn}\sum_{i=1}^{vn} X_i < \tfrac{1}{v} \right ) < e^{-cn}.
\ee
This inequality contradicts the consequence of the Ballistic Assumption proved in Theorem \ref{T1}.

\end{proof}

\section{The expectation of $|\gamma|$ is infinite} \label{secT2}
The object of this section is to prove the following theorem, which, combined with Corollary \ref{collsec2}, contradicts the Ballistic Assumption.

\begin{theorem} \label{T2}
The following holds:
\be \label{Enoninf}
\E_{\iRB}(|\gamma|) = \infty.
\ee
\end{theorem}
Theorem \ref{T2} will be proven by contradiction. Let us suppose that there exists a constant $\nu$ such that
\be \label{T2not}
\E_{\iRB}(|\gamma|) < \nu< \infty.
\ee
 
The main tool in this section is the operation of stickbreaking.

\begin{opr}
A renewal time of a bridge $\gamma(i) \in R_\gamma$ is called a diamond time of $\gamma$ if the two first coordinates of all other points of $\gamma$  lie in the cone 
\bee \label{cone}
\{(x,y): x >x(\gamma(i)), x-x(\gamma(i)) \ge y-y(\gamma(i)) > -(x-x(\gamma(i))) \}
\cup \nonumber \\
\{(x,y): x < x(\gamma(i)), x(\gamma(i))-x > y-y(\gamma(i)) ge -(x(\gamma(i))-x) \}.  \label{cone}
\eee
\par The set of all diamond times of the bridge $\gamma$ put in increasing order will be denoted by $D_\gamma$.
\end{opr}

\par The set $D_\gamma$ has a positive density in $R_\gamma$ under the assumptions \eqref{BA} and \eqref{T2not}.
\begin{lemma} \label{diamonddensity}
Suppose that $\E_{\iRB}(|\gamma|)  <\infty$. Then, there exists $\delta > 0$ such that 
\be																																						
\liminf_{n \to \infty} \frac{D_\gamma  \cap [0,n]}{n} \ge \delta
\ee
almost surely.
\end{lemma}

\begin{proof}
The probability measure $\p_{\iRB}$ is invariant under reflection $\mathcal R _x(\gamma)$. Therefore the expectation of the $y$-coordinate of the endpoint of $\gamma \in \iRB$ is equal to zero.
\par The finite expectation of $|\gamma|$ also implies that 
\be
\E_{\iRB}(\max_{t \in \gamma}|y(t)|) \le \E_{\iRB}(|\gamma|) <\infty.
\ee
\par We can apply the law of large numbers to find that for $\gamma$ distributed accordingly to $\p_{\iRB}^{\otimes \N}$, there exists a constant $\mu >0$ such that 
\be \label{ibend}
\left( \frac{x(\gamma(r_n))}{n}, \frac{y(\gamma(r_n))}{n} \right) \to (\mu,0) \text{ almost surely.}
\ee
This result implies that there exists a positive integer $K$ and a non-zero probability $p_K$ for $\gamma$ to lie in a half-cone $\{ v \in \Z^d, (x(v)+K) \ge y(v) > -(x(v)+K)\}$. Indeed,
\be \label{coneK}
\p_{\iRB}^{\otimes \N} \left( \left\{\inf_{i \ge 0} \left(x(\gamma(i))+y(\gamma(i))\right) \ge -K\right\} \cap \left\{\inf_{i \ge 0} \left(x(\gamma(i))-y(\gamma(i))\right) \ge -K\right\} \right) \ge p_K.
\ee
\par Taking any step of the distribution $\rho$ and applying all necessary reflections and turns, we can obtain a one step walk ${\tilde{\gamma}}$ such that $x(\tilde{\gamma}(1)) \ge y(\tilde{\gamma}(1)) \ge 0$.
Then, $\gamma_0 = \tilde{\gamma} \circ \mathcal R_x (\tilde{\gamma})$ is located in a cone $\{v\in \Z^d, x(v) \ge y(v) > -x(v)\}$ and the end of $\gamma_0$ lies on the hyperplane $\{v \in Z^d, y(v)=0\}$. The weight of the segment equivalent to $\gamma_0$ will be denoted by $\sigma_0$.
\par By construction of $\p_{\iRB}^{\otimes \N}$, we can add $K$ samples of $\gamma_0$ to the beginning of any infinite walk $\gamma$. If $\gamma$ lies in a cone $\{ v \in \Z^d, (x(v)+K) \ge y(v) > -(x(v)+K)\}$ used in \eqref{coneK}, then the result of this addition will be located in a cone $\{v\in \Z^d, x(v) \ge y(v) > -x(v)\}$.
The probability price of this operation is equal to $\sigma_0^K$.
\par We can combine this fact with \eqref{coneK} to obtain that
\be 
\p_{\iRB}^{\otimes \N} \left( \gamma(0) \in D_\gamma \right) \ge \sigma_0^K \,p_K.
\ee
\par The same bound is true for the bi-infinite random bridge:
\be 
\p_{\iRB}^{\otimes \Z} \left( \gamma(0) \in D_\gamma \right) =\delta \ge (\sigma_0^K \,p_K)^2.
\ee 
By the invariance of $\p_{\iRB}^{\otimes \Z}$ under the operation of shift, 
$\p_{\iRB}^{\otimes \Z} \left( \gamma(r_k) \in D_\gamma \right) =\delta$ for any $r_k \in R_\gamma $.
The estimated density of diamond points is then equal to 
$$\lim_{n \to \infty} \E_{\iRB}^{\otimes \Z}\left( \frac{D_\gamma \cap \{r_k\}_{k=0}^n}{n} \right) = \delta.$$
We can use this fact and apply Lemma \ref{tauerg} to the shift-invariant event 
$$\liminf_{n \to \infty} \left| \frac{D_\gamma \cap \{r_k\}_{k=0}^n}{n} \right| \ge \delta$$ 
to conclude that it has probability equal to 1.
\end{proof}

\par The definition of the diamond point can be extended to the bridges of finite length. It is easy to see that if $\gamma_n \in \RB_n$ coincides with the beginning of $\gamma$, then $D_{\gamma} \cap \gamma_n \subset D_{\gamma_n}$.
The operation consisting in taking the finite part of the bridge can only add new diamond points  but not destroy the initial ones. For bridges $\gamma$ with at least two diamond points, we can define the operation of stickbreaking.

\begin{opr}[Stickbreaking]
Suppose that $\gamma \in \RB_n$ and that there exist two points $i, j \in D_{\gamma}$ with $i<j$. Then, define a new bridge via the formula:
$$\StBr_{(i,j)}(\gamma) = (\gamma(k))_{k=0}^i \circ r_{\pi/2} \big((\gamma(k))_{k=i}^j\big) \circ(\gamma(k))_{k=j}^n$$. 
\end{opr}
This operation does not add any crossing to the walk, so the weight does not change: $\sigma(\gamma)=\sigma(\StBr_{(i,j)}(\gamma))$ for any choice of diamond points $i$ and $j$. Also, note that the result of this operation is not necessary a bridge.

\begin{proof}
\par Let us assume \eqref{T2not}. From any infinite bridge $\gamma \in \B_\infty$, we can take a finite beginning containing the first $n$ irreducible bridges of the walk: $\gamma^{(n)} = (\gamma(i))_{i=0}^{r_n}$. Let us use the notation $\tilde{\gamma} \triangleleft \gamma$ to say that there exists a renewal point $r_n \in R_\gamma$ such that $\tilde\gamma = \gamma^{(n)}$. 

\par Define the width of any finite bridge as follows:
\be
W(\gamma) = \max_{0 \le i,j,\le |\gamma|} \left( y(\gamma(i)) - y(\gamma(j)) \right).
\ee

\par Now, fix $\e > 0$ (the exact value of the constant $\e$ will be determined later). Look at the set of infinite bridges starting with not very long and not very wide finite bridges:
\be \label{WNbridge}
\overline{\RB_\infty^+}(n,\e) = \left\{ \gamma \in \RB_\infty^+: 
|\gamma^{(n)}| < \nu n,
W(\gamma^{(n)}) < \e n,
|D_{\gamma^{(n)}}| \ge \frac{\delta n}{2}\right\}.
\ee
The exact value of the constant $\e$ will be determined later.

The irreducible bridges that form $\gamma \in \RB_\infty^+$ are independent and identically distributed so we can use the law of large numbers and the formula \eqref{ibend} to conclude that 
\be
\lim_{n \to \infty}\p_{\iRB}^{\otimes \N}\left(\gamma: r_n < \nu n, W(\gamma^{(n)}) < \e n\right) =1.
\ee
The probability of the condition on the number of diamond points is the result of Lemma \ref{diamonddensity}:
\be
\lim_{n \to \infty}\p_{\iRB}^{\otimes \N} \left(\gamma: |D_{\gamma^{(n)}}| \ge \frac{\delta n}{2}\right) =1.
\ee
The combination of these two estimations gives us
\be \label{narrowprob1}
\lim_{n \to \infty} \p_{\iRB}^{\otimes \N} \left(\overline{\RB_\infty^+}(n,\e) \right) =1.
\ee
We obtain the contradiction with \eqref{narrowprob1} and prove the theorem by constructing the necessary amount of wide bridges using the operation of stickbreaking.

\par Let us define the set of all appropriate finite opening bridges as follows:
\be
(\overline{\RB_\infty^+}(n,\e))^{(n)} = \left\{
\tilde{\gamma}: \exists \gamma \in \overline{\RB_\infty^+}(n,\e), \tilde{\gamma} =\gamma^{(n)} 
\right\} .
\ee
Then, use Lemma \ref{infcondfin} to estimate the probability of this set in the following way:
\be
c \left((\overline{\RB_\infty^+}(n,\e))^{(n)}\right)=
\p_{\iRB}^{\otimes \N} \left(\overline{\RB_\infty^+}(n,\e)\right) =
\sum_{\tilde{\gamma} \in (\overline{\RB_\infty^+}(n,\e))^{(n)}} \mu_c^{-|\tilde\gamma|} \sigma(\tilde\gamma).
\ee

\par Let us take the bridge $\tilde\gamma \in (\overline{\RB_\infty^+}(n,\e))^{(n)}$ and the diamond points $d_i, d_j \in D_{\tilde\gamma}$ with $i \in [\frac{\delta n}{10}, \frac{2\delta n}{10}]$ and $j \in [\frac{3 \delta n}{10}, \frac{4\delta n}{10}]$. The result of the stickbreaking operation $\phi = \StBr_{d_i, d_j}(\tilde\gamma)$ is a bridge if the following conditions hold:
\begin{align}
& \min_{d_i \le k \le d_j} x(\phi (k)) > 0,  \label{tobeabridge1}
\\
& \max_{d_i \le k \le d_j} x(\phi (k)) \le x(\phi(|\tilde{\gamma}|)).\label{tobeabridge2}
\end{align}
Inequalities \eqref{tobeabridge1} and \eqref{tobeabridge2} are true if 
\be \label{tobeabridgecond}
x(\tilde\gamma(d_i)) - W((\tilde\gamma(k))_{k=d_i}^{d_j}) > \frac{\delta n}{10} - \e n >0.
\ee
To guarantee that the above is valid, choose for example $\e= \frac{\delta}{20}$.
\par The width of the result can be bound in the following way:
\be
W(\StBr_{d_i, d_j}(\tilde\gamma)) \ge \frac{\delta n}{10}\ge \e n.
\ee

\par The number of renewal points of $\phi = \StBr_{d_i, d_j}(\tilde\gamma)$ has the following upper bound:
\bal
|R_\phi| 
& = |R_{\tilde\gamma}| + |R_\phi \cap \{\phi(k)\}_{k=d_i}^{d_j}| - |R_{\tilde\gamma} \cap \{\tilde\gamma(k)\}_{k=d_i}^{d_j}| \nonumber \\
& \le n + W\left((\tilde\gamma(k))_{k=d_i}^{d_j}\right) - |j-i| \nonumber \\
& \le n+ \e n - \frac{\delta n}{10} \le n.
\end{align}
We can conclude that any $\gamma \in \B_{\infty}^+$ starting with $\phi$ does not belong to $\overline{\B_\infty^+}(n,\frac{\delta}{20})$ because $W(\gamma^{(n)}) > \e n$.

\par The length of $\tilde\gamma$ cannot be bigger than $\nu n$. Hence, the number of $\tilde\gamma$ that can form the beginning of some fixed $\gamma \in \RB_{\infty}^+$ after the stickbreaking with some choice of $i$ and $j$ can be bounded by the number of ways to choose $i$ and $j$ over $\nu n$ possibilities
\be \label{stbr-1card}
\left| \left\{
(\tilde\gamma, i, j) \in (\overline{\RB_\infty^+}(n,\e))^{(n)} \times [\tfrac{\delta n}{10}, \tfrac{2\delta n}{10}] \times [\tfrac{3 \delta n}{10}, \tfrac{4\delta n}{10}]: 
\StBr_{d_i, d_j}(\tilde\gamma) \triangleleft \gamma
\right\}\right| \le (\nu n)^2.
\ee
For any fixed choice of $\tilde{\gamma} \in (\overline{\RB_\infty^+}(n,\e))^{(n)}, \, i \in [\frac{\delta n}{10}, \frac{2\delta n}{10}]$ and $j \in [\frac{3 \delta n}{10}, \frac{4\delta n}{10}]$, Lemma \ref{infcondfin} implies that for $\phi = \StBr_{d_i, d_j}(\tilde\gamma)$,
\be \label{weight1}
\p_{\iRB}^{\otimes \N}(\exists \gamma \in \RB_\infty^+:  \phi \triangleleft \gamma) = 
e^{-\lambda_0 |\phi|} \sigma(\phi) = e^{-\lambda_0 |\tilde{\gamma}|} \sigma(\tilde{\gamma}) .
\ee
After the summing  of \eqref{weight1} over all possible $\tilde\gamma, i$ and $j$, and plugging the sum in \eqref{stbr-1card}, we obtain that
\bal
\left(\frac{\delta n}{10}\right)^2 \sum_{\tilde\gamma \in (\overline{\RB_\infty^+}(n,\e))^{(n)} }  e^{-\lambda_0 |\tilde{\gamma}|} \sigma(\tilde{\gamma})
&= \left(\frac{\delta n}{10}\right)^2\p_{\iRB}^{\otimes \N} (\overline{\RB_\infty^+}(n,\e)) \nonumber \\
&\le (\nu n)^2 \, \p_{\iRB}^{\otimes \N} \left( \gamma: \exists (\tilde\gamma, i, j): \StBr_{d_i, d_j}(\tilde\gamma) \triangleleft \gamma \right) \nonumber \\
&\le (\nu n)^2 \, \p_{\iRB}^{\otimes \N} (W(\gamma^{(n)}) > \e n ).
\end{align}
 This inequality contradicts \eqref{narrowprob1}, so the assumption \eqref{T2not} has to be rejected.

\end{proof}

\paragraph{Acknowledgments}  
The author thanks Hugo Duminil-Copin for posing this problem and reading the text.
The work is supported by the Swiss FNS and the NCCR Swissmap.

\bibliographystyle{alpha}

\end {document}